\documentclass[a4paper,USenglish,cleveref,autoref,thm-restate]{lipics-v2021}

\usepackage{amsmath,amsfonts,amssymb,amsthm}
\usepackage{mathbbol}
\usepackage{graphicx,color}
\usepackage{boxedminipage}
\usepackage[linesnumbered, boxed, algosection,noline,noend]{algorithm2e}
\usepackage{framed}
\usepackage{thmtools}
\usepackage{xspace}
\usepackage{todonotes}
\usepackage{pgfplots}
\usepackage{xifthen}
\usepackage{tabularx}
\usepackage{capt-of}
\usepackage{subcaption}
\usepackage{booktabs}
\usepackage{calc}
\usepackage{xargs}

\newcommandx{\ktodo}[2][1=]{\todo[textcolor=red,linecolor=red,backgroundcolor=red!25,bordercolor=red,#1]{K: #2}}

\usetikzlibrary{calc}
\usetikzlibrary{shapes}
\usetikzlibrary{arrows.meta}

\pgfplotsset{compat=1.18}

\DeclareMathOperator{\operatorClassETH}{{\sf ETH}\xspace}
\newcommand{\classETH}{\ensuremath{\operatorClassETH}\xspace}

\DeclareMathOperator{\operatorClassNP}{{\sf NP}}
\newcommand{\classNP}{\ensuremath{\operatorClassNP}}

 \DeclareMathOperator{\poly}{poly}

\newlength{\RoundedBoxWidth}
\newsavebox{\GrayRoundedBox}
\newenvironment{GrayBox}[1]%
   {\setlength{\RoundedBoxWidth}{.93\textwidth}
    \def\boxheading{#1}
    \begin{lrbox}{\GrayRoundedBox}
       \begin{minipage}{\RoundedBoxWidth}}%
   {   \end{minipage}
    \end{lrbox}
    \begin{center}
    \begin{tikzpicture}%
       \node(Text)[draw=black!20,fill=white,rounded corners,%
             inner sep=2ex,text width=\RoundedBoxWidth]%
             {\usebox{\GrayRoundedBox}};
        \coordinate(x) at (current bounding box.north west);
        \node [draw=white,rectangle,inner sep=3pt,anchor=north west,fill=white] 
        at ($(x)+(6pt,.75em)$) {\boxheading};
    \end{tikzpicture}
    \end{center}}     

\newenvironment{defproblemx}[2][]{\noindent\ignorespaces%
                                \FrameSep=6pt%
                                \parindent=0pt%
                \vspace*{-1.5em}
                \ifthenelse{\isempty{#1}}{%
                  \begin{GrayBox}{\textsc{#2}}%
                }{%
                  \begin{GrayBox}{\textsc{#2} parameterized by~{#1}}%
                }
                \begin{tabular*}{\textwidth}{@{\hspace{.1em}} >{\itshape} p{1.8cm} p{0.8\textwidth} @{}}%
            }{
                \end{tabular*}%
                \end{GrayBox}%
                \ignorespacesafterend
            }

\newcommand{\defproblema}[3]{
  \begin{defproblemx}{#1}
    Input:  & #2 \\
    Task: & #3
  \end{defproblemx}
}

\newcommand{\pname}{\textsc}
\newcommand{\ProblemFormat}[1]{\pname{#1}}
\newcommand{\ProblemIndex}[1]{\index{problem!\ProblemFormat{#1}}}
\newcommand{\ProblemName}[1]{\ProblemFormat{#1}\ProblemIndex{#1}{}\xspace}

\newcommand{\probDiameter}{\ProblemName{$k$-Diameter}}

\newcommand{\probColoring}{\ProblemName{$k$-Coloring}}
\newcommand{\probThreeColoring}{\ProblemName{$3$-Coloring}}
\newcommand{\probgridembedded}{\ProblemName{Grid Embedded SAT}}
\newcommand{\probcliqueparition}{\ProblemName{Clique Cover}}
\newcommand{\probThreeCP}{\ProblemName{$3$-Clique Cover}}

\pdfoutput=1 
\hideLIPIcs  

\title{Subexponential Algorithms for Clique Cover on Unit Disk and Unit Ball Graphs}

\author{Tomohiro Koana}{Utrecht University, the Netherlands}{tomohiro.koana@gmail.com}{https://orcid.org/0000-0002-8684-0611}{Supported by the European Research Council (ERC) under the European Union’s Horizon 2020 research and innovation programme (project CRACKNP under grant agreement No. 853234).}
\author{Nidhi Purohit}{National University of Singapore, Singapore}{nidhipurohit95@gmail.com}{https://orcid.org/0000-0003-4869-0031}{}
\author{Kirill Simonov}{Hasso Plattner Institute, University of Potsdam, Germany}{kirillsimonov@gmail.com}{https://orcid.org/0000-0001-9436-7310}{}

\authorrunning{T. Koana, N. Purohit and K. Simonov} 

\Copyright{Tomohiro Koana, Nidhi Purohit and Kirill Simonov} 



\category{} 

\relatedversion{} 

\nolinenumbers 

\begin{document}

\maketitle

\begin{abstract}
    In \probcliqueparition, given a graph $G$ and an integer $k$, the task is to partition the vertices of $G$ into $k$ cliques.
    \probcliqueparition on unit ball graphs has a natural interpretation as a clustering problem, where the objective function is the maximum diameter of a cluster.

    Many classical \classNP-hard problems are known to admit $2^{O(n^{1 - 1/d})}$-time algorithms on unit ball graphs in $\mathbb{R}^d$ [de Berg et al., SIAM J. Comp 2018]. A notable exception is the \textsc{Maximum Clique} problem, which admits a polynomial-time algorithm on unit disk graphs and a subexponential algorithm on unit ball graphs in $\mathbb{R}^3$, but no subexponential algorithm on unit ball graphs in dimensions $4$ or larger, assuming the ETH [Bonamy et al., JACM 2021].

    In this work, we show that \probcliqueparition also suffers from a ``curse of dimensionality'', albeit in a significantly different way compared to \textsc{Maximum Clique}. We present a $2^{O(\sqrt{n})}$-time algorithm for unit disk graphs and argue that it is tight under the ETH. On the other hand, we show that \probcliqueparition does not admit a $2^{o(n)}$-time algorithm on unit ball graphs in dimension $5$, unless the ETH fails.
\end{abstract}

\newpage 

\section{Introduction} \label{sec:intro}

Clustering is a general method of partitioning data entries, normally represented by points in the Euclidean space, into clusters with the goal of minimizing a certain similarity function for the points in the same cluster. Many popular similarity objectives such as $k$-means and $k$-center are center-based, i.e., the objective function of the cluster is defined in terms of distance to the additionally selected center of the cluster. On the other hand, arguably the most natural similarity measure that is defined solely in terms of distances between the given datapoints, is the maximum diameter of a cluster. That is, the objective function of the clustering is the maximum distance between any pair of points in the same cluster. Formally, we consider the following 
\probDiameter problem: Given a set of points $P$ in the Euclidean space $\mathbb{R}^d$, and parameters $k$, $D$, is there a partitioning of $P$ into disjoint $C_1$, \ldots, $C_k$, such that for each $j \in [k]$, and each $x, y \in C_j$, $||x - y|| \le D$?\footnote{Since one can binary search over the value of $D$, and there are at most $|P|^2$ different distances between the pairs of points, this decision version of the problem is equivalent to the optimization version, up to logarithmic factors in the running time.}

The \probDiameter problem admits a natural geometric interpretation. Consider a set of disks with centers in $P$ and of the same radius $D / 2$. The problem asks to partition the disks into $k$ sets so that disks in each set pairwise intersect.
Given a graph $G$ and an integer $k$, let \probcliqueparition be the problem of partitioning the vertex set of $G$ into $k$ vertex-disjoint cliques.
\probDiameter in $\mathbb{R}^d$ is thus equivalent to \probcliqueparition on unit ball graphs in~$\mathbb{R}^d$. Note that \probcliqueparition is equivalent to \textsc{$k$-Coloring} on general graphs by taking the complement of the graph; however, unit ball graphs are not closed under complements, therefore \probcliqueparition on unit ball graphs does not necessarily have the same complexity as \textsc{$k$-Coloring} on unit ball graphs.


The main question we ask in this work is the following: Does \probDiameter in $\mathbb{R}^d$, or equivalently \probcliqueparition on $d$-dimensional unit ball graphs, admit subexponential-time algorithms?
Given that \probcliqueparition is a natural graph problem akin to \textsc{Maximum Clique} and \textsc{$k$-Coloring}, our question fits into the recent line of advances for algorithms on geometric intersection graphs.

In a seminal work, de Berg et al.~\cite{BergBKMZ20} gave a framework for $2^{O(n^{1 - 1/d})}$-time algorithms on, in particular, $d$-dimensional unit ball graphs, which covers problems such as \textsc{Maximum Independent Set}, \textsc{Dominating Set}, and \textsc{Steiner Tree}. At the heart of the framework lies a special kind of tree decomposition, that essentially guarantees that each bag is covered by $O(n^{1 - 1/d} / \log n)$ cliques. The target problem is then solved via dynamic programming over the decomposition, given that the interaction of the solution with the cliques in the bag could be succinctly represented. For example, in the \textsc{Maximum Independent Set} problem the solution can have at most one element per clique, and storing the intersection between the solution and the bag is therefore sufficient for the running time above.

However, \probcliqueparition stands aside from the problems covered by the framework of de Berg et al., as the interaction between the smallest clique cover and the given clique cover of the bag does not immediately seem to admit a succinct representation.
Moreover, one can easily observe that finding the smallest clique cover is still \classNP-hard even if a clique cover of the graph of constant size is given. Indeed, it is famously \classNP-hard to determine whether a 4-colorable graph admits a 3-coloring~\cite{KhannaLS00, GuruswamiK04}, and colorings turn into clique covers under taking the complement of the graph.

Previously in the literature, another problem shown to not exhibit such a ``gradually subexponential'' behavior was the \textsc{Maximum Clique} problem. Already since the 1990s, a polynomial-time algorithm for \textsc{Maximum Clique} on unit disk graphs was known~\cite{ClarkCJ90}. Recently, Bonamy et al.~\cite{BonamyBBCGKRST21} have shown that \textsc{Maximum Clique} only admits a subexponential-time algorithm on $3$-dimensional unit ball graphs, while no $2^{o(n)}$-time algorithm is possible in dimension 4, assuming the \classETH.

\subparagraph*{Our results.} As the first step, we show a subexponential algorithm for \probcliqueparition on unit disk graphs.
Our starting point is the weighted treewidth approach of de Berg et al.~\cite{BergBKMZ20}; however, as per the discussion above, on its own this characterization does not seem to be sufficient. Intuitively, the geometric structure of unit disk graphs has to play a role not only in the decomposition itself, but also in representing the solution with respect to the decomposition.
In order to accommodate this, we build upon the classical lemma due to Capoyleas, Rote and Woeginger~\cite{CapoyleasRW91}, that was rediscovered several times in the literature~\cite{DumitrescuP11, PirwaniS12}.
Simply put, there always exists an optimal clique cover where all cliques are well-separated, i.e., the convex hulls of the respective disk centers do not intersect. As only constantly many cliques may lie in direct vicinity of another clique in an optimal solution, we can show that there are at most polynomially many possible configurations for each clique in the optimal solution.
This characterization, coupled with the dynamic programming approach, results in the following theorem.

\begin{restatable}{theorem}{subexpalgo}
    \label{theorem:subexpalgo}
    \probcliqueparition can be solved in time $2^{O(\sqrt n)}$ on $n$-vertex unit disk graphs, when a geometric representation of the graph is given in the input, with bit-length of the vectors bounded by $\poly(n)$.
\end{restatable}
Note that recognizing unit disk graphs is, in general, \classNP-hard~\cite{BreuK98} and even $\exists\mathbb{R}$-complete~\cite{KangM12}, which means that one cannot expect to be able to compute a geometric representation of a given unit disk graph efficiently.

Using the lower bound machinery of de Berg et al.~\cite{BergBKMZ20}, we also observe that the running time above is tight.
Moreover, the lower bound holds for higher dimensions as well.

\begin{restatable}{theorem}{subexplb}
    Assuming the \classETH, \probcliqueparition on $n$-vertex unit ball graphs in $\mathbb{R}^d$ does not admit a $2^{o(n^{1 - 1/d})}$-time algorithm, for any $d > 1$, even if the geometric representation of polynomial bit-length is given in the input.
    \label{thm:eth_Rd}
\end{restatable}

The next natural question is whether the algorithmic result of Theorem~\ref{theorem:subexpalgo} could also be extended to higher dimensions. Unfortunately, the separation property that plays the key role in Theorem~\ref{theorem:subexpalgo} only holds in the two-dimensional case: the original work of Capoyleas, Rote and Woeginger already observes that the analogous statement in three dimensions admits a counterexample~\cite{CapoyleasRW91}. This, however does not exclude other potential ways for a succinct representation of the solution, or another completely unrelated approach. We show that the separation property is indeed crucial, that is, \probcliqueparition does not admit subexponential algorithms on unit ball graphs in constant dimension.

\begin{restatable}{theorem}{explb}
    Assuming the \classETH, \probcliqueparition on $n$-vertex unit ball graphs in $\mathbb{R}^5$ does not admit a $2^{o(n)}$-time algorithm, even if the geometric representation of polynomial bit-length is given in the input.
    \label{thm:eth_R5}
\end{restatable}

To put Theorem~\ref{thm:eth_R5} into context, recall the result of Bonamy et al.~\cite{BonamyBBCGKRST21}, showing that \textsc{Maximum Clique} does not admit a subexponential algorithm on unit ball graphs in~$\mathbb{R}^4$. Their approach is to first argue that \textsc{Maximum Independent Set} is as hard on 2-subdivisions (graphs obtained by replacing each edge with a path of length 3) as it is on general graphs, which holds simply because a maximum independent set of a graph can be extracted from a maximum independent set of its 2-subdivision. Then their key structural observation is that a complement of \emph{any} 2-subdivision admits a unit ball representation in $\mathbb{R}^4$, therefore showing hardness of \textsc{Maximum Clique} on unit ball graphs in $\mathbb{R}^4$. Note that \textsc{Maximum Independent Set} turns into \textsc{Maximum Clique} by taking the complement.

Since we target the \probcliqueparition problem on unit ball graphs, a natural idea is to conduct the reduction in a similar spirit, but starting from \probColoring. However, the obstacle is that 2-subdivisions do not in general preserve the existence of a $k$-coloring --- only for $k = 2$, which is not suitable for a hardness reduction. Therefore, instead of replacing each edge by its 2-subdivision, we need to use a more complicated edge gadget, and the 4-dimensional representation of Bonamy et al. is no longer applicable. The straightforward triangle-like edge gadget that preserves $3$-colorings could be used in place of the 2-subdivision, see Figure~\ref{fig:gadget} for an illustration. However, it is not clear whether the resulting graph would admit a sufficiently low-dimensional representation, namely below dimension $7$. Instead, the gadget that we use is based on two parallel 2-subdivisions, plus special vertices that impose a list-coloring-like condition on the internal vertices of the subdivisions; this choice of the gadget allows us to decrease the dimension to $5$ (see Figure~\ref{fig:gadget} for the illustration of the gadget).

\begin{figure}[h]
    \centering
    \includegraphics[width=\textwidth]{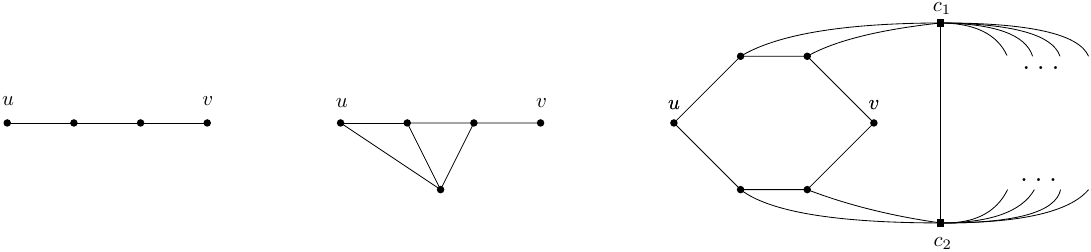}
    \caption{Edge gadgets encoding the edge between vertices $u, v$: left, $2$-subdivision of the edge, suitable for maximum independent sets; center, triangle-like gadget suitable for $3$-colorings; right, improved gadget preserving $3$-colorings --- here vertices $c_1$ and $c_2$ are connected in the same way to all edge gadgets.}
    \label{fig:gadget}
\end{figure}

\section{Preliminaries}
\label{sec:prelim}

\paragraph*{Sets, vectors and coordinates.}
For an integer $n$, we use $[n]$ to denote the set $\{1, 2, \ldots, n\}$.
We use the tuple notation for points in $\mathbb{R}^d$, i.e., a point is defined by the tuple $(a_1, a_2, \ldots, a_d)$, where $a_i \in \mathbb{R}$ is the respective coordinate for each $i \in [d]$. The variables $x_1$, $x_2$, \ldots, $x_d$ are used to denote the respective axes. We denote the origin by $O = (0, 0, \ldots, 0)$, and by $Ox_ix_j$, $i, j \in [d]$ we denote the plane spanned on the respective axes; the same notation is used for higher-dimensional subspaces too. For two points $A, B \in \mathbb{R}^d$, $\overrightarrow{AB}$ denotes the vector pointing from $A$ to $B$, its coordinates are expressed as $B - A$. We use $|| \cdot ||$ to denote the standard Euclidean norm in $\mathbb{R}^2$, therefore, $||B - A||$ is the Euclidean distance between the points $A$ and $B$, and also the length of the vector $\overrightarrow{AB}$.

\paragraph*{Unit ball graphs.}
Let $P = \{ p_1, \dots, p_n \}$ be a set of points in $\mathbb{R}^d$ and $B$ be a set of balls $b_i$ of radius $1$, centered at $p_i$.
A unit ball graph on $P$ is a graph over the vertex set $P$, in which two vertices $p_i$ and $p_j$ are adjacent if and only if the balls $b_i$ and $b_j$ intersect.

\paragraph*{Exponential-time hypothesis.}
The exponential-time hypothesis (\classETH), due to Impagliazzo, Paturi and Zane~\cite{ImpagliazzoP01, ImpagliazzoPZ01}, implies that there is no algorithm that solves \textsc{3-SAT} in $2^{o(n)}$ time, where $n$ is the number of variables in the formula. Since by the Sparsification Lemma~\cite{ImpagliazzoPZ01} this holds even for linearly-many clauses in the formula, \classETH also excludes $2^{o(n + m)}$-time algorithms for \textsc{3-SAT}, where $m$ is the number of clauses. By the standard linear-size reduction from \textsc{3-SAT} to \textsc{3-Coloring}, \classETH implies that \textsc{3-Coloring} does not admit a $2^{o(n + m)}$-time algorithm, where $n$ is the number of vertices and $m$ is the number of edges in the graph.

\paragraph*{Tree decomposition.}

For a graph $G = (V, E)$, a \emph{tree decomposition} is a pair $(T, \sigma)$, where $T = (V_T, E_T)$ is a tree and $\sigma \colon V_T \to 2^{V}$ such that
\begin{itemize}
    \item for each $uv \in E$, there exists $t \in V_T$ with $u, v \in \sigma(t)$, and 
    \item for each $v \in V$, the set of nodes $t \in V_T$ with $v \in \sigma(t)$ forms a connected subtree in $T$.
\end{itemize}
The \emph{width} of $(T, B)$ is $\max_{x \in V(T)}(|B(x)| - 1)$.
The \emph{tree-width} of $G$ is the minimum width of all tree decompositions of $G$.

A \emph{nice} tree decomposition is a tree decomposition more amenable to the design of dynamic programming algorithms.
Formally, a tree decomposition $(T = (V_T, E_T), \sigma)$ rooted at $r \in V_T$ is called \emph{nice} if $\sigma(r) = \emptyset$ and each node $t \in V_T$ is one of the following types:
\begin{description}
    \item[Leaf node.] $t$ is leaf in $T$ and $\sigma(t) = \emptyset$.
    \item[Introduce node.] $t$ has exactly one child $t'$, and $\sigma(t) = \sigma(t') \cup \{ v \}$ for a vertex $v$ in $G$.
    \item[Forget node.] $t$ has exactly one child $t'$, and $\sigma(t) = \sigma(t') \setminus \{ v \}$ for a vertex $v$ in $G$.
    \item[Join node.] $t$ has exactly two children $t', t''$, and $\sigma(t) = \sigma(t') = \sigma(t'')$. 
\end{description}
It is known that given a tree decomposition, a nice tree decomposition of the same width can be computed in polynomial time \cite{Bodlaender96}.

\section{Subexponential algorithm for unit disks}

In this section, we design a  subexponential-time algorithm for \probcliqueparition on unit disk graphs.

\subexpalgo*

To design a subexponential-time algorithm, let us introduce two known techniques.
We start with the ``separation theorem'' of Capoyleas, Rote and Woeginger~\cite{CapoyleasRW91}.
Recall that we aim to partition the vertex set of a given unit disk graph into a collection of $k$ cliques.
Each clique is defined by the convex hull of the centers of disks in the clique.
In principle, these convex hulls may arbitrarily intersect each other.
The following states that we may assume that they are disjoint in an optimal solution.

\begin{theorem}[Capoyleas, Rote and Woeginger~\cite{CapoyleasRW91}]
    \label{prop:separation}
    For \probcliqueparition on unit disk graphs, there exists an optimal solution $(C_1, \dots, C_{\ell})$ such that the convex hulls of the centers in $C_i$ are pairwise disjoint.
\end{theorem}

This was first proven by Capoyleas, Rote and Woeginger~\cite{CapoyleasRW91} but also by Dumitrescu and Pach~\cite{DumitrescuP11} and Pirwani and Salavatipour~\cite{PirwaniS12} later.
Theorem~\ref{prop:separation} relies crucially on the fact that for two intersecting convex polygons $P_1, P_2$ of diameter at most $d$, there exists two disjoint convex polygons $P_1', P_2'$ of diameter at most $d$ such that the vertices of $P_1$ and $P_2$ are contained in $P_1' \cup P_2'$.
In view of Theorem~\ref{prop:separation}, we will show that there are polynomially many ``relevant'' cliques in Lemma~\ref{lemma:clique-bound}.
To prove this, we will also use the following simple fact.

\begin{lemma}[Dumitresku and Pach, Lemma 2~\cite{DumitrescuP11}]
    \label{lemma:area-cover}
    Let $(G, \ell)$ be an instance of \probcliqueparition on unit disk graphs, and $(C_1, \dots, C_{\ell})$ be an optimal solution satisfying the condition of Theorem~\ref{prop:separation}.
    For a set $S$ of vertices contained in a square of constant side length, there are $O(1)$ cliques $C_i$ that intersect $S$.
\end{lemma}

See Dumitrescu and Pach~\cite{DumitrescuP11} for a concrete bound in the above lemma.
Now we prove a polynomial bound on the number of relevant cliques.

\begin{lemma}
    \label{lemma:clique-bound}
    Let $(G, \ell)$ be an instance of \probcliqueparition on unit disk graphs.
    Given $S \subseteq V(G)$,
    we can find in polynomial time a collection $\mathcal{R}$ of cliques in $G$ such that $|\mathcal{R}| \in |S|^{O(1)}$, and for each optimal solution $(C_1, \dots, C_{\ell})$ satisfying the condition of Theorem~\ref{prop:separation}, $S \cap C_i \in \mathcal{R}$ for all $i \in [\ell]$.
\end{lemma}
\begin{proof}
    Let $C = S \cap C_i$ be a clique with $C \ne \emptyset$.
    We will say that a clique $C_j$ is \emph{close} to $C$ if their closest vertices have distance at most two and \emph{far} from $C$ otherwise.

    We first show how to separate $C$ from far cliques, i.e., we find a collection of closed regions $P$ such that $C$ lies within $P$ and any far clique lies outside $P$.
    Suppose that $u, v \in C$ are two vertices with the largest distance $r \le 2$ in $C$.
    Then, $C$ is contained in the intersection of two disks of radius $r$ centered at $u$ and $v$, and every vertex of every far clique from $C$ is outside of these disks.
    For each $u, v \in S$, let $P_{u, v}$ be the intersection of such two disks, and let $R_{u, v}$ be the vertices of $S$ that lie in $P_{u, v}$.
    Let $\mathcal{R}'$ be the collection of vertex sets containing $R_{u, v}$ for each $u, v \in S$.
    We then have $|\mathcal{R}'| \in O(|S|^2)$, and for every $C = S \cap C_i$ there exists $R \in \mathcal{R}'$ that does not intersect any clique far from $C$.

    Next, we discuss how to separate $C$ from close cliques.
    By the above characterization, $C$ is contained in a $2 \times 4$-rectangle (not necessarily axis-aligned).
    For each close clique $C'$ of $C$, there exists a vertex $t$ in $C'$ with distance at most $2$ to a vertex in $C$, and every vertex in $C'$ has distance at most $2$ to $t$, so every vertex of $C'$ is at most at distance $4$ from some vertex of $C$.
    Therefore by extending the $2 \times 4$ rectangle containing $C$ by $4$ in every direction, we obtain a $10 \times 12$ rectangle that contains every close clique of $C$.
    Thus, by Lemma~\ref{lemma:area-cover}, there are $O(1)$ close cliques $C_j$ with $j \in [\ell]$.
    For a clique $C_j$, $j \in [\ell]$, since the convex hulls of $C$ and $C_j$ do not overlap by Theorem~\ref{prop:separation}, there is a line that separates $C$ and $C_j$, this line also separates the convex hulls of $C$ and $C_j \cap S$.
    Moving this line, we find two vertices on the boundary of the convex hull of $C$ or two vertices on on the boundary of the convex hull of $S \cap C_j$, such that the line through them separates $C$ and $S \cap C_j$ in the plane.
    Let $\mathcal{P}''$ be the collection of regions
    obtained as the intersection of constantly\footnote{The constant depends on Lemma~\ref{lemma:area-cover}.} many open or closed semi-planes whose boundaries go through two points of $S$.
    Let $\mathcal{R}''$ be the collection of vertex sets such that for each region in $\mathcal{P}''$, there is a vertex set in $\mathcal{R}''$ containing exactly the vertices of $S$ lying in this region.


    Finally, let $\mathcal{R}$ be the collection of intersections of $R'$ and $R''$ for $R' \in \mathcal{R}'$ and $R'' \in \mathcal{R}''$.
    Clearly, $|\mathcal{R}| \in |S|^{O(1)}$. By the above, we have that for $C = S \cap C_i$, 
    there exists $R' \in \mathcal{R}'$ that is disjoint from $S \cap C_j$ for every clique $C_j$ that is far from $C$, and there exists $R'' \in \mathcal{R}''$ that is disjoint from $S \cap C_h$ for every clique $C_h$ close to $C$; on the other hand, $R'$ and $R''$ contain $C$.
    Therefore, $R' \cap R''$ is disjoint from $S \cap C_j$ for every $j \ne i$, and contains $C$.
    Since $V(G) = C_1 \cup \ldots \cup C_\ell$ and $R', R'' \subseteq S \subseteq V(G)$, $R' \cap R''$ contains no vertices outside of $C$, and $C = R' \cap R'' \in \mathcal{R}$.
    Moreover, every $R \in \mathcal{R}$ is a clique since every $R' \in \mathcal{R}'$ is a clique.
    This completes the proof of the lemma.
\end{proof}

We will also use the framework of de Berg et al.~\cite{BergBKMZ20} for the design of subexponential-time algorithms for geometric intersection graphs.
First, let us introduce some terminology.
For a graph $G = (V, E)$ and $\kappa \in \mathbb{N}$, a \emph{$\kappa$-partition} of $G$ is a partition $(P_1, \dots, P_{\nu})$ of $V$ such that every $P_i$ induces a connected subgraph which is a union of at most $\kappa$ cliques.
For a $\kappa$-partition $\mathcal{P}$ of $G$, the \emph{$\mathcal{P}$-contraction} of $G$, denoted by $G_{\mathcal{P}}$, is the graph obtained by contracting every $P_i$ into a single vertex, that is, $V(G_{\mathcal{P}}) = \{ P_1, \dots, P_{\nu} \}$ and $E(G_{\mathcal{P}}) = \{ P_i P_j \mid \exists v_i \in P_i, v_j \in P_j \colon v_i v_j \in E(G)\}$.
Let $\gamma \colon \mathbb{N} \to \mathbb{N}$ be a weight function.
For a tree decomposition $(T, \sigma)$ of $G_{\mathcal{P}}$, its \emph{weighted width} with respect to $\gamma$ is defined by $\max_{t} \sum_{P_i \in \sigma(t)} \gamma(|P_i|)$, where the maximum is over the nodes $t$ of $T$.

The main technical step of the algorithmic framework of de Berg et al.\ is the following theorem, restricted to the case of unit disk graphs.

\begin{theorem}[\cite{BergBKMZ20}, Theorem 2.11 applied to unit disk graphs]\label{thm:framework}
    For a weight function $\gamma$ such that $\gamma(t) \in O(t^{1/2-\varepsilon})$ for $\varepsilon > 0$, there exists a $\kappa$-partition $\mathcal{P}$ for $\kappa \in O(1)$ such that $G_{\mathcal{P}}$ has weighted treewidth $O(\sqrt{n})$ that can be computed in $2^{O(\sqrt{n})}$ time.
\end{theorem}

As in Berg et al.~\cite{BergBKMZ20}, we will apply Theorem~\ref{thm:framework} with $\gamma(t) = O(\log t)$.
To design a $2^{O(\sqrt{n})}$-time algorithm, one essentially needs to show that there are $|P|^{O(1)}$ possibilities for each partition class $P \in \mathcal{P}$.
We obtain this polynomial bound from Lemma~\ref{lemma:clique-bound}.
Specifically, let $(C_1, \ldots, C_\ell)$ be an optimal solution satisfying the condition of Theorem~\ref{prop:separation}.
For $P \in \mathcal{P}$, let $\mathcal{R}(P)$ be the collection of cliques returned by Lemma~\ref{lemma:clique-bound}, applied to the subset $P \subseteq V(G)$.
By the lemma, for every $i \in [\ell]$, $P \cap C_i \in \mathcal{R}(P)$.
On the other hand, every clique is contained in a $2 \times 4$ rectangle, therefore by Lemma~\ref{lemma:area-cover}
only constantly many cliqes from $C_1$, \ldots, $C_\ell$ intersect this clique.
Since $P$ is covered by at most $\kappa$ cliques, 
it also holds that only constantly many cliqes from $C_1$, \ldots, $C_\ell$ intersect $P$, where the constant depends on $\kappa$ and Lemma~\ref{lemma:area-cover}; denote this constant by $\lambda$.
Later in the algorithm, we will characterize the solution $(C_1, \ldots, C_\ell)$ on $P$ by listing the $\lambda$ cliques from $\mathcal{R}(P)$ that result from intersecting $(C_1, \ldots, C_\ell)$ with $P$.
We now proceed to the proof of the theorem.

\begin{proof}[Proof of Theorem~\ref{theorem:subexpalgo}]
    We first apply Theorem~\ref{thm:framework} with $\gamma(t) = \varepsilon \log t + 1$ for a sufficiently small constant $\varepsilon > 0$, obtaining a $\kappa$-partition $\mathcal{P}$ of $G$, and a tree decomposition of $G_{\mathcal{P}}$ of weight at most $O(\sqrt{n})$.
For $P \in \mathcal{P}$, let $\mathcal{R}(P)$ be a collection of relevant cliques in $P$ as per Lemma~\ref{lemma:clique-bound}.
We define a \emph{configuration} of $P$ by a pair $(\mathcal{C}, \chi)$ as follows.
The first element, $\mathcal{C} \subseteq \mathcal{R}(P)$, is a collection of at most $\lambda$ cliques such that $\bigcup \mathcal{C} = P$.
The second element, $\chi \colon \mathcal{C} \to \{ 0, 1 \}$, is a mapping, which we will use to indicate whether a clique $C \in \mathcal{C}$ has been covered.
We denote the set of configurations of $P$ by $\Gamma_P$.
Since $|\mathcal{R}(P)| \in |P|^{O(1)}$ by Lemma~\ref{lemma:clique-bound}, there are at most $\lambda \cdot |\mathcal{R}(P)|^{\lambda} \cdot 2^{\lambda} \in |P|^{O(1)}$ many configurations.
Thus, for a bag $t$, the number of all combinations of configurations of nodes in $t$ is at most
\begin{align*}
    \prod_{P \in \sigma(t)} |P|^{O(1)} = \exp \left( c \sum_{P \in \sigma(t)} \log |P| \right) \in 2^{O(\sqrt{n})}.
\end{align*}
Here, $c$ is a constant, and the second equality is due to the fact that the weighted treewidth is $O(\sqrt{n})$.
The running time will be dominated by this factor.

Our dynamic programming constructs a table $c_t$ for a bag $t$ indexed by a configuration for each $P \in \sigma(t)$ and an integer $\ell$.
We describe the configuration by a mapping $f$ that maps $P \in \sigma(t)$ to one of its configurations in $\Gamma_P$.
We use the notation $f(P) = (f_{\mathcal{C}}(P), f_{\chi}(P))$.
The table $c_t$ stores Boolean values, where the entry  $c_t[f, \ell]$ is true  if and only if there is a collection $(C_1, \dots, C_{\ell})$ of $\ell$ cliques such that
\begin{itemize}
    \item
        $\bigcup_{i \in [\ell]} C_i$ covers all vertices appearing strictly below $t$ (i.e., every vertex in $P \in \mathcal{P} \setminus \sigma(t)$ such that $P$ appears in the subtree rooted at $t$ is covered by $\bigcup_{i \in [\ell]} C_i$)
    \item
    $\bigcup_{i \in [\ell]} C_i$ covers all cliques $C \in f_{\mathcal{C}}(P)$ with $P \in \sigma(t)$ and $f_{\chi}(P)(C) = 1$, and 
    \item
    every clique $C_i$, $i \in [\ell]$, contains a vertex appearing strictly below $t$.
\end{itemize}
Our dynamic programming will maintain this invariant.

Now we describe our dynamic programming procedure over a nice tree decomposition (see Section~\ref{sec:prelim} for the definition).
It follows from our invariants that the input graph admits a clique cover of size $\ell$ if and only if $c_r[f, \ell] = 1$ for the root $r$.
For a non-leaf node $t$, we will denote its children by $t', t''$ ($t'$ if $t$ has one child).

\proofsubparagraph{Leaf node.}
Suppose that $t$ is a leaf node, i.e., $\sigma(t) = \emptyset$.
Then, $c_t[f, \ell]$ is true if and only if $\ell = 0$.

\proofsubparagraph{Introduce node.}
Suppose that $t$ is an introduce node, i.e., $\sigma(t) = \sigma(t') \cup \{ P \}$.
\begin{align*}
    c_t[f, \ell] =
    \begin{cases}
        c_{t'} [f|_{\sigma(t')}, \ell] & \text{ if  $f_{\chi}(P)(C) = 0$ for every $C \in f_{\mathcal{C}}(P)$}, \\
        \text{false} & \text{ otherwise.}\\
    \end{cases}
\end{align*}
Here, $f|_{\sigma(t')}$ denotes the restriction of $f$ to $\sigma(t')$.
As we only consider cliques that intersect a node strictly below $t$, we set the table entry to false if $f_{\chi}(P)$ is not uniformly zero.

\proofsubparagraph{Forget node.}
Suppose that $t$ is a forget node, i.e., $\sigma(t) = \sigma(t') \setminus \{ P \}$.
We have the following recurrence:
\begin{align*}
    c_t[f, \ell] = \bigvee_{\ell' \in \{ 0, \dots, \ell \}, f'} c_{t'}[f', \ell']
\end{align*}
where $\bigvee$ ranges over all $f'$ and $\ell'$ that satisfy the following condition.
Let $H_{f'}$ be an auxiliary graph as follows.
For every $C \in f_{\mathcal{C}}'(P)$ with $f_{\chi}(P)(C) = 0$, we add a vertex $h_C$.
Moreover, for every $P' \in \sigma(t)$ and $C \in f_{\mathcal{C}}(P')$, we add a vertex $h_{C}$ to $H$ if (i) $C$ has not been covered at $t'$, i.e., $f_{\chi}'(P')(C) = 0$ and (ii) $C$ is covered at $t$, i.e., $f_{\chi}(P')(C) = 1$.
Two vertices $h_{C}$ and $h_{C'}$ are adjacent in $H_{f'}$ if and only if $C \cup C'$ is a clique in the graph $G$.
This concludes the construction of $H_{f'}$.
Note that $H_{f'}$ has size $O(\sqrt{n})$.
Then $\bigvee$ ranges over $f'$ and $\ell'$ such that $H_{f'}$ has a clique cover $(D_{1}, \dots, D_{\ell - \ell'})$ of size $\ell - \ell'$ such that every clique $D_i$ contains a vertex $h_{C}$ for $C \in f_{\mathcal{C}}'(P)$.
Whether $f'$ and $\ell'$ fulfills this condition can be checked in $2^{O(\sqrt{n})}$ time via dynamic programming.

Specifically, we proceed in a standard fashion for a \probColoring/\probcliqueparition subset-based dynamic programming.
For each subset $S \subseteq V(H_{f'})$ and each integer $k$, $0 \le k \le \ell - \ell'$, we compute the Boolean value $d[S, k]$
that is equal to true if and only if the subgraph $H_{f'}[S]$ admits a clique cover of size $k$, where additionally every clique contains a vertex $h_C$ for some $C \in f'_C(P)$.
We initialize by setting $d[\emptyset, 0] = $ true, $d[S, 0] = $ false for each $S \ne \emptyset$, and for each $S \subseteq V(H_{f'})$, $k \in [\ell - \ell']$,
compute $d[S, k] = \bigvee_{D \text{ is an admissible clique in } H_{f'}[S]} d[S \setminus D, k - 1]$.
Clearly, the dynamic programming table above is computed in time $2^{O(|V(H_{f'})|)} = 2^{O(\sqrt{n})}$.
As there are $2^{O(\sqrt{n})}$ many choices for the configuration $f'$, we can compute $c_t[f, \ell]$ in overall time $2^{O(\sqrt{n})}$.

Let us verify that the invariant is maintained by the computation above.
If $c_t[f, \ell]$ is set to true, then there exist $f'$ and $\ell'$ satisfying the aforementioned condition, for which $c_{t'}[f', \ell']$ is also true.
Since $c_{t'}[f', \ell']$ is true, there exists a collection $(C_1, \dots, C_{\ell'})$ of cliques.
Also, $H_{f'}$ admits clique cover of size $\ell - \ell'$, which is also a collection of cliques in $G$.
Combining these cliques indeed satisfies the conditions.
 
\proofsubparagraph{Join node.}
Suppose that $t$ is a join node, i.e., $\sigma(t) = \sigma(t') = \sigma(t')$.
We have the recurrence:
\begin{align*}
    c_t[f, \ell] = \bigvee_{\ell' \in \{ 0, \dots, \ell \},\, f'\!,\, f''} (c_{t'}[f', \ell'] \wedge c_{t''}[f'', \ell - \ell']), 
\end{align*}
where $\bigvee$ ranges over functions $f', f''$ that map $P \in \sigma(t)$ to one of its configurations such that for every $P \in \sigma(t)$,
\begin{itemize}
    \item $P$ is partitioned in cliques in the same way, i.e., $f_{\mathcal{C}}(P) = f_{\mathcal{C}}'(P) = f_{\mathcal{C}}''(P)$, and 
    \item for every $C \in f_{\mathcal{C}}(P)$, $f_{\mathcal{\chi}}(P)(C) = 1$ if and only if $C$ is covered in one of the children, i.e., $f_{\mathcal{\chi}}'(P)(C) = 1$ or $f_{\mathcal{\chi}}''(P)(C) = 1$.
\end{itemize}

To see why the invariant is maintained, note that if $c_t[f, \ell]$ is set to true, then there are $\ell'$ cliques certifying $c_{t'}[f', \ell']$ being true and $\ell - \ell'$ cliques certifying $c_{t''}[f'', \ell - \ell']$ being true.
Putting them together, we obtain a collection of $\ell$ cliques satisfying the conditions.

Observe that each entry can be computed in $2^{O(\sqrt{n})}$ time.
Since there are $2^{O(\sqrt{n})}$ entries, the running time is bounded by $2^{O(\sqrt n)}$.
Note that all arithmetic operations can be performed in polynomial time: we only require comparing distances between the given points and orientations between triples of given points; see the proof of Lemma~\ref{lemma:clique-bound}.
Theorem~\ref{thm:framework} is representation-agnostic, meaning that no additional arithmetic operations are required, except for constructing the graph from the given geometric representation.

This concludes the proof of Theorem~\ref{theorem:subexpalgo}.
\end{proof}

\section{Subexponential lower bound for $d \ge 2$}
In this section, we establish the impossibility of solving the \probcliqueparition problem on $d$-dimensional unit ball graphs in time better than $2^{O(n^{1 - 1/d})}$.
For this, we use the result of de Berg et al.~\cite{BergBKMZ20}, which states that, assuming \classETH, \probgridembedded cannot be solved in time $2^{o(n)}$ time, where $n$ is the number of variables of the given formula.

\probgridembedded is defined as follows.
Let $G^2(n)$ denote the $n \times n$-grid graph, where there is a vertex $(i, j)$ for every $i, j \in [n]$ and an edge between $(i,j)$ and $(i', j')$ are adjacent if and only if $|i - i'| = |j - j'| = 1$.
We say that a graph $H$ is \emph{embedded} in $G^2(n)$ if a subdivision of $H$ is isomorphic to a subgraph of $G^2(n)$.
For a CNF formula $\phi$, the \emph{incidence graph} $G_{\phi}$ of $\phi$ is the bipartite graph, where there is a vertex for each variable and each clause, and there is an edge between a variable vertex and a clause vertex if and only if the variable appears in the clause.
A $(3, 3)$-CNF formula is a CNF formula where each variable appears at most 3 times and each clause has size at most 3.


\defproblema{\probgridembedded}%
{A $(3,3)$-CNF formula $\phi$ together with an embedding of its incidence graph $G_{\phi}$ in $G^2({n})$.}%
{Is there a satisfying assignment for $\phi$?}
\begin{proposition}[\cite{BergBKMZ20}, Theorem 3.2]{\label{prop:gridembeddedeth}}
    \probgridembedded can not be solved in time $2^{o(n)}$ unless \classETH fails.
 \end{proposition}

To show \classETH-hardness for $\mathbb{R}^d$, $d \ge 3$, we use the \emph{cube wiring theorem} due to de Berg et al.~\cite{BergBKMZ20}.
Let $B^{d}(n)$ denote $[n]^d$ and $G^d(n)$ denote the $d$-dimensional hypercube over $B^d(n)$.
Also, for $p \in B^{d-1}(n)$ and $h \in [n]$, let $\xi^{h}(p) = (p_1, \dots, p_{d-1}, h) \in B^{d}(n)$.
For $s \in \mathbb{N}$, a set $P \subseteq \mathbb{Z}^{d-1}$ is said to be $s$-spaced if there is an integer $0 \le r < s$ such that for every $p = (p_1, \dots, p_{d-1}) \in P$ and $i \in [d-1]$, $p_i \equiv r \bmod 2$.

\begin{theorem}[Cube wiring theorem \cite{BergBKMZ20}]\label{thm:cube}
    For $d \ge 3$, let $P$ and $Q$ be two 2-spaced subsets of $B^{d-1}(n)$ and let $M$ be a perfect matching in the bipartite graph $(P \cup Q, P \times Q)$.
    Then, for $n' \in O(n)$, $G^d(n')$ contains vertex-disjoint paths that connect $\xi^{1}(p)$ and $\xi^{n'}(q)$ for every $pq \in M$.
\end{theorem}

Now we prove our theorem.
 
\subexplb*

\begin{proof}
    We first present a reduction from \probgridembedded to \probcliqueparition on unit disk graphs. 
    Let $\phi$ be a $(3,3)$-CNF formula.
    We may assume that each variable in $\phi$ appears twice positively and once negatively:
    For every variable $v$ where its occurrences are all positive or negative, delete the clauses containing~$v$.
    Also, for every variable $v$ appears twice negatively and once positively, flip its sign.
    We first describe how to construct a \probgridembedded instance $(G, k)$ from $\phi$, and specify the embedding later.
 \begin{itemize}
     \item  For each variable $x$, we create a \emph{variable gadget}, which is obtained by gluing $K_3$ and $K_2$ over one vertex, i.e., it is a paw, consisting of four vertices $u_x, u_x', v_x, w_x$ and edges $u_x u_x', u_x v_x, u_x'v_x, v_x w_x$.
     We will call $u_x, u_x', w_x$ \emph{connection vertices}.
     \item  For each clause, we introduce a single vertex $C$.
        We call it a \emph{clause} gadget.
  \item   We construct a \emph{wire} gadget, which will be used to connect a variable gadget to a clause gadget in the embedding.
  A wire corresponding to a positive literal $x$ is a path with an even number of edges, starting at $u_x$ or $u_x'$ from the variable gadget of $x$, and ending at the corresponding clause vertex.
  For a negative literal, the path starts at $w_x$ instead.
  We call a wire \emph{activated} if the value of the corresponding literal of the connection vertex is true.
    We call a wire if the corresponding literal is set to true.
 \end{itemize}
 This completes the construction of $G$. 
 Let $L$ be the total edge length of all wires. 
We show that the formula $\phi$ has a satisfying assignment if and only if $G$ has a clique cover of size $k = n + L/2$, where $n$ is the number of variables.

\proofsubparagraph{Correctness.}

Suppose that formula $\phi$ has a satisfying assignment.
We construct a clique cover of $G$ as follows.
   For each variable $x$, we pick a clique $\{ u_x, u_x', v_x \}$ if $x$ is assigned true and $\{ v_x, w_x \}$ otherwise.
   For each wire with $2\ell$ edges, pick $\ell$ edges as $K_2$'s so that all inner vertices and the connection vertex is covered if the wire is activated, and all inner vertices and the clause vertex is covered otherwise.
   Since the assignment satisfies all the clauses, every clause gadget has at least one activated wire.
   If more than one wire ends with $K_2$ containing a clause, then arbitrarily pick one wire and reduce the internal vertices of the remaining wires into a $K_1$ and pick it into the solution.
    Hence, all vertices in the variable, wire and clause gadget are covered.
We obtain a clique partition of $G$ with $k= n + L / 2$ cliques.

Conversely, assume $G$ has a clique cover of size $k$.
Each variable gadget contains at least one clique that covers the common vertices of the gadget.
Since in a wire of length $2\ell$, there are $2\ell-1$ internal vertices, and only two vertices of a wire can be covered by a clique. Thus, wire gadgets contain at least $L / 2$ cliques. 
Since $k=n+L / 2$, the solution contains exactly one clique for every variable gadget and each wire of length $2\ell$ will have exactly $\ell$ cliques.
Since every clause vertex belongs to a clique in the solution, a literal exists such that the corresponding wire is activated.
Then, the respective connection vertex is not a part of the wire clusters and thus is a part of the vertex cluster.
We assign the variable's value based on which side the clique in each variable gadget picks.
If the clique picks connection vertices corresponding to $K_3$, we set the variable to be true. If the clique contains connection vertices corresponding to $K_2$, then we set the variable to be false.
Otherwise, we set variable values arbitrarily.

\proofsubparagraph{Embedding.}

Suppose that $d = 2$.
Let $\mathcal{D}$ be a grid embedding of $G_{\phi}$.
We start by taking a 2-refinement of $\mathcal{D}$.
This will ensure that each wire gadget has even length.
For every vertex in $G_{\phi}$, we introduce a disk (of diameter 1) centered at its coordinate, unless it is a variable vertex. 
For a variable $x$, let $(i, j)$ be its coordinate in $\mathcal{D}$.
Without loss generality, assume that three vertices adjacent to $x$ in $G_{\phi}$ are at $(i -1, j)$, $(i, j + 1)$, and $(i+1, j)$.
There are three cases depending on which edge in $G_{\phi}$ incident with $x$ connects to a negative literal.

First, suppose that the edge between $(i,j)$ and $(i, j+1)$ leads to a negative literal.
Then, introduce four disks centered at $(i-1/2, j-1/2), (i+1/2, j-1/2)$ (corresponding to $u_x$ and $u_x'$), $(i, j-1/2)$ (corresponding to $v_x$), and $(i,j+1/2)$ (corresponding to $w_x$).
Otherwise, suppose that the between $(i,j)$ and $(i+1, j)$ leads to a negative literal.
Then, introduce four disks centered at $(i-1/2, j), (i, j +1/2)$ (corresponding to $u_x$ and $u_x'$), $(i, j-1/2)$ (corresponding to $v_x$), and $(i+1/2, j-1/2)$ (corresponding to $w_x$).
See Figure~\ref{fig:emb} for an illustration.

Note that only polynomial precision in coordinates is required to construct the instance, therefore the hardness also holds if the representation is given.

For $d \ge 3$, for every variable, we place three vertices adjacent to its variable gadget in a $(d-1)$-hypercube of side length 3.
We then place all these hypercubes into a $(d-1)$-hypercube of side length $n^{O(\frac{1}{d-1})}$.
Placing the clause gadgets on $B^{d-1}(m^{\frac{1}{d-1}})$, we apply the cube wiring theorem (Theorem~\ref{thm:cube}) to obtain an embedding into $B^{d}(n')$ for $n' \in O(n)$.
We then embed the variable gadgets similarly to the case $d = 2$.
\end{proof}

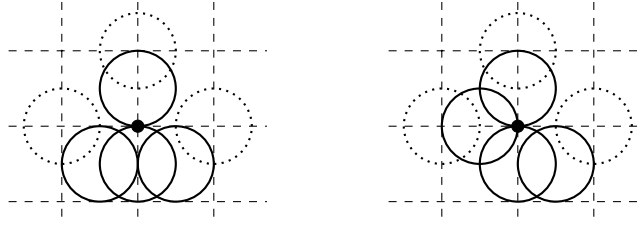
\begin{figure}
    \center
    \begin{tikzpicture}
        \draw[dashed] (-1.7,-1.2) grid (1.7, 1.7);
        \draw[fill] (0,0) circle (.08);
        \draw[thick] (-0.5, -0.5) circle (.5);
        \draw[thick] (0, -0.5) circle (.5);
        \draw[thick] (0.5, -0.5) circle (.5);
        \draw[thick] (0, 0.5) circle (.5);
        \draw[thick,dotted] (-1, 0) circle (.5);
        \draw[thick,dotted] (1, 0) circle (.5);
        \draw[thick,dotted] (0, 1) circle (.5);

        \begin{scope}[shift={(5, 0)}]
            \draw[dashed] (-1.7,-1.2) grid (1.7, 1.7);
        \draw[fill] (0,0) circle (.08);
            \draw[thick,dotted] (-1, 0) circle (.5);
            \draw[thick,dotted] (1, 0) circle (.5);
            \draw[thick,dotted] (0, 1) circle (.5);
            \draw[thick] (-0.5, 0) circle (.5);
            \draw[thick] (0, 0.5) circle (.5);
            \draw[thick] (0, -0.5) circle (.5);
            \draw[thick] (0.5, -0.5) circle (.5);
        \end{scope}
    \end{tikzpicture}
    \caption{Two cases for a variable gadget. The coordinate $(i, j)$ is marked by the black dot. The dotted disks are part of wire gadgets. Note that the variable gadget has exactly one disk intersecting a dotted disk.}
    \label{fig:emb}
\end{figure}

\section{Exponential lower bound for $d = 5$}\label{sec:r5}

In this section, we present a hardness reduction excluding better-than-exponential running time for \probcliqueparition on unit ball graphs in dimension at least $5$.
We restate the result next.

\explb*
\begin{proof}
    We show a reduction from \probThreeColoring to \probcliqueparition, where the target instance is a unit ball graph in $\mathbb{R}^d$.
    Let $G$ be the graph in the instance of \probThreeColoring.
    We first construct an enhanced graph $G'$ from $G$ and argue that this makes an equivalent instance of \probThreeColoring.
    Then, we show that the complement of the enhanced graph $G'$ admits a unit ball representation in $\mathbb{R}^5$.
    Since solving \probThreeColoring on $G'$ is equivalent to solving \probThreeCP on $\overline{G'}$, and \probThreeCP on unit ball graphs in $\mathbb{R}^5$ is the special case of \probcliqueparition with $k = 3$ on the same class of graphs, this completes the reduction.

    We now move to the details. First, we define the enhanced graph $G'$. The vertex set of $G'$ contains one vertex for each vertex of $G$, four vertices for each edge of $G$, and two additional special vertices. Formally, $V(G') = W \cup T \cup B \cup C$, where $W = \{w_v : v \in V(G)\}$, $T = \{t^1_e, t^2_e : e \in E(G)\}$, $B = \{b^1_e, b^2_e : e \in E(G)\}$, $C = \{c_1, c_2\}$. The edges are as follows: for every edge $e = uv \in E(G)$,    we construct  $w_ut_e^1$, $t_e^1t_e^2$, $t_e^2w_v$, and $w_ub_e^1$, $b_e^1b_e^2$, $b_e^2w_v$
    Additionally, $c_1$ is adjacent to all vertices of $T$, $c_2$ is adjacent to all vertices of $B$, and $c_1$ and $c_2$ are adjacent. Formally, the edge set of $G'$ is
    $$E(G') = \{w_ut_e^1, t_e^1t_e^2, t_e^2w_v, w_ub_e^1, b_e^1t_e^2, b_e^2w_v, t_e^1c_1, t_e^2c_1, b_e^1c_2, b_e^2c_2 : e \in E(G)\} \cup \{c_1c_2\}.$$
    Intuitively, $G'$ is obtained from $G$ by replacing each edge $e \in E(G)$ with two copies of its 2-subdivision: the vertices $t_e^1$ and $t_e^2$ are internal vertices of the first copy, and  the vertices $b_e^1$ and $b_e^2$ are internal vertices of the first copy. Moreover, there are two special vertices $c_1$ and $c_2$ that are adjacent to each other, and $c_1$ is adjacent to the internal vertices of the first 2-subdivision, while $c_2$ is adjacent to internal vertices of the second 2-subdivision. See Figure~\ref{fig:gadget2} for an illustration of the edge gadget.
\begin{figure}[h]
    \centering
    \includegraphics{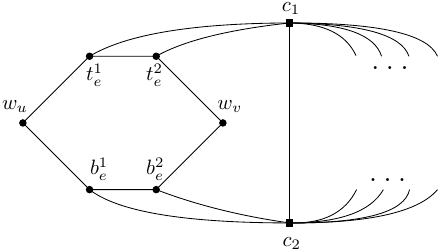}
    \caption{Edge gadget in $G'$, encoding the edge $e$ between vertices $u, v$ in $G$. Vertices $c_1$ and $c_2$ are connected in the same way to all edge gadgets.}
    \label{fig:gadget2}
\end{figure}

    We now argue that $G'$ is equivalent to $G$ in terms of $3$-colorings.
    \begin{claim}
        $G$ admits a $3$-coloring if and only if $G'$ admits a $3$-coloring.
    \end{claim}
    \begin{proof}
        Let $c: V(G) \to \{1, 2, 3\}$ be the $3$-coloring of $G$, we construct a $3$-coloring $c'$ of $G'$. Let $c'$ coincide with $c$ on the vertices of $W$; let $c'(c_1) = 1$ and $c'(c_2) = 2$.
        We now assign colors to vertices $t_e^h$ and $b_e^h$ for $e \in E(G), h \in [2]$.

        Consider an edge $e = uv \in E(G)$, so that $u$ is adjacent to $t_e^1$ and $b_e^1$ in $G'$. The vertex $t_e^1$ has an available color since only $c_1$ and $u$ have assigned colors among its neighbors; assign this color to $t_e^1$. Now, assume there is no available color for $t_e^2$, therefore all three colors appear among $c_1$, $v$, $t_e^1$. Since $c'(c_1) = 1$, either $c'(v) = 2$ and $c'(t_e^1) = 3$, or the other way around. In the former case, $c'(u) \ne 2$ since $c(\cdot)$ is a proper $3$-coloring of $G$. Assign $c'(t_e^1) = 2$ and $c'(t_e^2) = 3$; all edges between the considered vertices are properly colored. In the alternative case, the argument is symmetric: $c'(v) = 3$ and $c'(u) \ne 3$; assign $c'(t_e^1) = 3$ and $c'(t_e^2)$. The argument for the vertices $b_e^1$ and $b_e^2$ is analogous.

        In the other direction, consider a $3$-coloring $c'$ of $G'$; we claim that the restriction $c$ of $c'$ to $V(G)$ is a proper $3$-coloring of $G$. Assume this is not the case, therefore there exists an edge $e = uv \in E(G)$ with $c'(u) = c(u) = c(v) = c'(v)$.
        Since $c_1$ and $c_2$ are adjacent in $G'$, they receive different colors under $c'$ and so either $c'(c_1) \ne c'(u)$ or $c'(c_2) \ne c'(u)$; w.l.o.g. assume the former case.
        The vertex $t_e^1$ has only one available color since it cannot coincide with $c'(u)$ and $c'(c_1)$, which are two distinct colors. Then the neighborhood of $t_e^2$ contains all three colors, since $c'(u) = c'(v)$. This contradicts the fact that $t_e^2$ is properly colored by $c'$.
    \end{proof}

    Then we proceed to construct a unit ball representation of the complement of $G'$ in $\mathbb{R}^5$. To this end, we describe the locations of all vertices in $G'$ under the embedding, and argue that the distance between the locations exceeds a certain value if and only if the respective pair of vertices is adjacent in $G'$.

    First, we embed the vertices of $T$, $B$ in $C$ in the first three dimensions, i.e., their images are always zero in coordinates $4$ and $5$. Then, we embed the vertices of $W$ in the other two dimensions, i.e., such the coordinates $1$--$3$ are zeroed out. Finally, we shift the embedding of $T$ and $B$ slightly to achieve the desired edges between $W$ and $T \cup B$.

    Let $\epsilon > 0$ be a constant to be defined later. We place $c_1$ and $c_2$ symmetrically across the origin at distance of $\sqrt{3} - \sqrt{2}/2  + \epsilon$ along the first coordinate; that is,
    \begin{align*}
        \pi(c_1) &= (\sqrt{3} - \sqrt{2}/2 + \epsilon, 0, 0, 0, 0),\\
        \pi(c_2) &= (-\sqrt{3} + \sqrt{2}/2 - \epsilon, 0, 0, 0, 0).
    \end{align*}
    We then position the set $T$ on the circumference of a circle with the center on the $Ox_1$ axis lying in the plane orthogonal to the axis, with radius $r = 1 + \epsilon'$, and such that its center is $\sqrt{2}/2 - \epsilon$ away from the origin towards $-\infty$. We shall define the precise value of $\epsilon'$ later. The points of $T_1$ occupy the ``top cap'' of the circumference, i.e., a small arc close to $x_2 = r$, and the points of $T_2$ occupy the ``bottom cap'', i.e., close to $x_2 = -r$. We aim that for each $e \in E(G)$, $t_e^1$ lies directly opposite to $t_e^2$, while the remaining points are sufficiently close to each of them. Let $E(G) = \{e_1, \ldots, e_m\}$, we position the points $t_{e_1}^1$, \dots, $t_{e_m}^1$ evenly along the arc starting from the ``top'' of the circle, such that the angle between the two consecutive points is always $\delta/m$, measured from the center of the circle. We then place the points $t_{e_1}^2$, \dots, $t_{e_m}^2$ similarly, directly opposite to their counterparts. We define the exact positions as follows:
    \begin{align*}
        \pi(t_{e_j}^1) &= (-\sqrt{2}/2 + \epsilon, r \cdot \cos(\delta \cdot j/m), r \cdot \sin(\delta \cdot j/m), 0, 0),\\
        \pi(t_{e_j}^2) &= (-\sqrt{2}/2 + \epsilon, -r \cdot \cos(\delta \cdot j/m), -r \cdot \sin(\delta \cdot j/m), 0, 0).
    \end{align*}

    The points of $B$ are positioned very similarly, except that they are placed in a circle placed opposite across the origin to the circle above, i.e., its center is the point $(\sqrt{2}/2, 0, 0, 0, 0)$. And the points of $B_1$ ($B_2$) are placed close to $x_3 = r$ ($x_3 = -r$). Formally, 
    \begin{align*}
        \pi(b_{e_j}^1) &= (\sqrt{2}/2 - \epsilon, -r \cdot \sin(\delta \cdot j/m), r \cdot \cos(\delta \cdot j/m), 0, 0),\\
        \pi(b_{e_j}^2) &= (\sqrt{2}/2 - \epsilon, r \cdot \sin(\delta \cdot j/m), -r \cdot \cos(\delta \cdot j/m), 0, 0).
    \end{align*}

    Note that the image of every point in $T \cup B$ is exactly $R_1 = \sqrt{(\sqrt{2}/2 - \epsilon)^2 + r^2}$ away from the origin.
    Assume $\epsilon$ is such that $R_1 < 4$, and let $R_2 = \sqrt{4 - R_1^2}$. We place the points of $W$ in the plane $Ox_4x_5$ exactly at the distance of $R_2$ from the origin. Namely, consider the circle in $Ox_4x_5$ centered at the origin with the radius of $R_2$. We place the points of $W = \{w_1, \ldots, w_n\}$ evenly along the circumference, such that the angle between consecutive points is exactly $\delta/n$:
    \begin{align*}
        \pi(w_i) &= (0, 0, 0, R_2 \cdot \cos(\delta \cdot i/n), R_2 \cdot \sin(\delta \cdot i/n)).
    \end{align*}
    See Figure~\ref{fig:embedding} for the illustration of the embedding $\pi$.

    \begin{figure}[h]
        \centering
        \includegraphics[width=\textwidth]{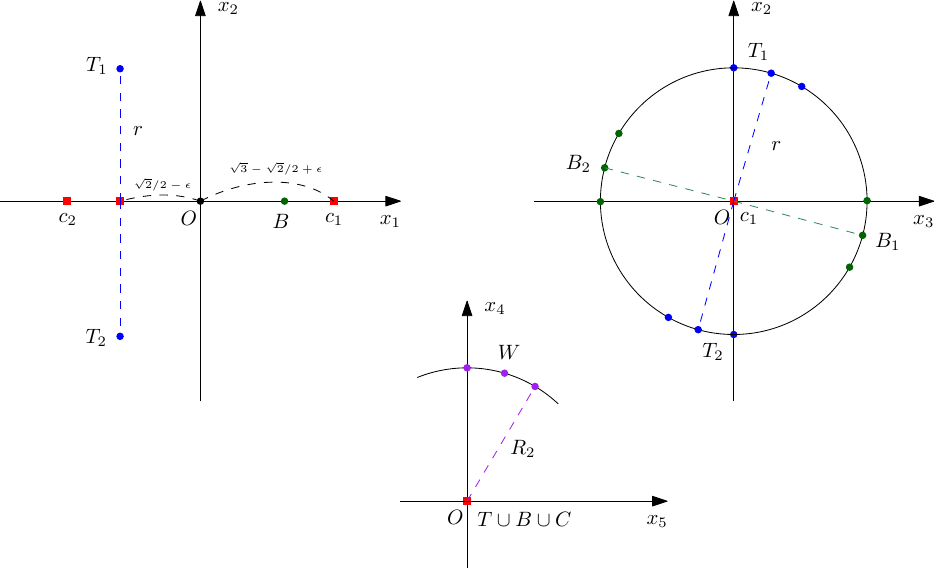}
        \caption{Illustration of the embedding $\pi$, showed by schematic projections on the three planes.}
        \label{fig:embedding}
    \end{figure}

    We now show that $\pi$ ``nearly'' gives the desired embedding of $G'$. That is, we show that every adjacent pair is at distance strictly more than $2$ and every non-adjacent pair is at distance strictly less than $2$, except for the pairs of form $(w, v)$, $w \in W$, $v \in T \cup B$, which are at distance exactly $2$. Later we will slightly modify the embedding $\pi$ to make sure that exactly the required pairs of this form are sufficiently far from each other.

    \begin{claim}
        There exists $\xi > 0$ such that the following holds:
        \begin{align}
            ||\pi(w) - \pi(v)|| &= 2, \text{ for each } w \in W, v \in T \cup B,\label{eq:pi1}\\
            ||\pi(t_e^1) - \pi(t_e^2)|| &\ge 2 + \xi, \text{ for each } e \in E(G),\label{eq:pi2}\\
            ||\pi(b_e^1) - \pi(b_e^2)|| &\ge 2 + \xi, \text{ for each } e \in E(G),\label{eq:pi3}\\
            ||\pi(t_e^h) - \pi(c_1)|| &\ge 2 + \xi, \text{ for each } e \in E(G), h \in [2],\label{eq:pi4}\\
            ||\pi(b_e^h) - \pi(c_2)|| &\ge 2 + \xi, \text{ for each } e \in E(G), h \in [2],\label{eq:pi5}\\
            ||\pi(c_1) - \pi(c_2)|| &\ge 2 + \xi,\label{eq:pi6}
        \end{align}
        and for any other two vertices $v, u$ of $G'$, the distance $||\pi(v) - \pi(u)||$ is at most $2 - \xi$.
        \label{claim:pi}
    \end{claim}
    \begin{proof}
        Let $\xi = \epsilon'/2$. Equation~\eqref{eq:pi1} holds immediately by construction, since each $v \in T \cup B$ is situated exactly $R_1$ away from the origin, each $w \in W$ exactly $R_2$ away from the origin, $T \cup B$ is contained in the 3-dimensional subspace $Ox_1x_2x_3$ which is orthogonal to the plane $Ox_4x_5$ where $W$ is contained, and $R_1^2 + R_2^2 = 4$ by definition of $R_2$.

        For Equation~\eqref{eq:pi2}, observe that $||\pi(t_e^1) - \pi(t_e^2)|| = 2 + 2\epsilon'$ for each $e \in E(G)$ since these two points are situated diametrically opposite to each other on a circle of radius $1 + \epsilon'$. Therefore, $||\pi(t_e^1) - \pi(t_e^2)|| \ge 2 + \xi$ since $\epsilon' = 2\xi \ge \xi/2$. On the other hand, consider the points $t_e^1$ and $t_{e'}^2$ for $e \ne e'$. Since $t_{e'}^2$ lies on the same circle at the angle of at least $\delta/m$ away from $t_e^2$, the distance between $t_e^1$ and $t_{e'}^2$ is at most $2 (1 + \epsilon') \cos (\delta/2m)$. Therefore, if it holds that $(1 + \epsilon') \cos(\delta/2m) \le 1 + \xi/2 = 1 + \epsilon'/4$, then all distances between the points of $T$ are as desired. Moreover, exactly the same arguments hold for Equation~\eqref{eq:pi3} and the distances between the points of $B$. We now show this bound given that $\epsilon'$ is sufficiently small:
        \begin{multline*}
            \epsilon' \le \frac{\delta^2}{20m^2} \implies 
            \frac{1 + \epsilon'}{\epsilon'} \ge \frac{20m^2}{\delta^2} \implies
            \frac{\delta^2}{16 m^2} \ge \frac{5}{4} \cdot \frac{\epsilon'}{1 + \epsilon'}\\
            \implies 1 - \frac{\delta^2}{16 m^2} \le 1 - \frac{5}{4} \cdot \frac{\epsilon'}{1 + \epsilon'} = \frac{1 - \epsilon'/4}{1 + \epsilon'}\\
            \implies
            (1 + \epsilon') \cdot \cos\frac{\delta}{2m} \le (1 + \epsilon') \cdot (1 - \frac{\delta^2}{16 m^2}) \le 1 - \epsilon'/4.
        \end{multline*}
        Here, we also use that $\delta$ is a sufficiently small constant.

        Consider now Equation~\eqref{eq:pi4}, the distance between $\pi(c_1)$ and $\pi(t_e^h)$ is equal to $\sqrt{3 + r^2}$ for each $e \in E(G)$ and $h \in [2]$. It is therefore sufficient to have $\sqrt{3 + (1 + \epsilon')^2} \ge 2 + \xi$, which holds since $\epsilon' = 2\xi$, and the same argument holds for Equation~\eqref{eq:pi5} because of the symmetry.
        Note that the distance between $\pi(c_2)$ and $\pi(t_e^h)$ for any $e \in E(G)$, $h \in [2]$ is equal to 
        $$\sqrt{(\sqrt{3} - \sqrt{2} + 2\epsilon)^2 + r^2} = \sqrt{(\sqrt{3} - \sqrt{2})^2 + 1 + O(\epsilon + \epsilon')} < 1.5 \le 2 - \xi,$$
        for sufficiently small $\epsilon$, $\epsilon'$ and $\xi$. The same holds for $\pi(c_1)$ and $\pi(b_e^h)$ for any $e \in E(G)$, $h \in [2]$.

        For Equation~\eqref{eq:pi6}, $$||\pi(c_1) - \pi(c_2)|| = 2\sqrt{3} - \sqrt{2} + 2\epsilon\ge 2.049 \ge 2 + \xi$$
        when $\xi$ is sufficiently small.

        It remains to verify that pairwise distances not discussed above are bounded by $2 - \xi$. Consider first $w \in W$ and $c_h$ for $h \in [2]$,  the respective squared distance is
        \begin{multline*}
        ||\pi(w) - \pi(c_h)||^2 = R_2^2 + (\sqrt{3} - \sqrt{2}/2 + \epsilon)^2 = 4 - R_1^2 + (\sqrt{3} - \sqrt{2}/2)^2 + O(\epsilon)\\
        \le 5.025 - R_1^2 + O(\epsilon) = 5.025 - (\sqrt{2}/2 - \epsilon)^2 - (1 + \epsilon')^2 + O(\epsilon) \\
        = 5.025 - 1/2 - 1 + O(\epsilon + \epsilon') = 3.525 + O(\epsilon + \epsilon') \le 3.8
        \end{multline*}
        for sufficiently small $\epsilon$ and $\epsilon'$. Therefore, $||\pi(w) - \pi(c_h)|| \le 2 - \xi$ when $\xi$ is sufficiently small.

        Note also that when $\delta$ is a sufficiently small constant, distances between all pairs of vertices in $W$ under $\pi$ are at most $1$, since the images occupy an arc which is a small fraction of a constant-radius circle, and the same holds for pairs in $T_1$, $T_2$, $B_1$, $B_2$.

        Finally, it remains to consider pairs of the form $t \in T$, $b \in B$. Observe that when projecting $T$ and $B$ orthogonally on the plane $Ox_2x_3$, these sets lie on the same circle of radius $r = 1 + \epsilon'$, and the radial distance between a point in $T$ and a point in $B$ is always at most $\pi / 2 + \delta$, since $T_1$, $B_1$, $T_2$, $B_2$ are each rotated $\pi/2$ further away from the previous set, and each of the four sets occupies an arc of radial length at most $\delta$. Therefore, 
        \begin{multline*}
            (2 - \xi)^2 - ||\pi(t) - \pi(b)||^2 \ge (2 - \xi)^2 - (\sqrt{2} - 2\epsilon)^2 - 2 (1 + \epsilon')^2 (1 + \sin \delta) \\
            \ge 4(\sqrt{2} - 1) \epsilon - 4\xi - 6 \epsilon' - 8 \delta,
        \end{multline*}
        by using $\epsilon', \epsilon \le 1$ and $\sin \delta \le \delta$.
        The above value is greater than zero if $\epsilon \ge 8(\epsilon' + \xi + \delta)$.

        To conclude the proof of the claim,  we note that the parameters $\xi$, $\epsilon'$, $\delta$, $\epsilon$ clearly admit values that satisfy all the restrictions above. Indeed, it is only required that each of them does not exceed a certain constant independent of the other parameters, and additionally that $2\xi = \epsilon' \le \frac{\delta^2}{20m^2}$, and $\epsilon \ge 8(\epsilon' + \xi + \delta)$.

    \end{proof}

    Finally, we construct the embedding $\pi'$ that gives the desired representation of $G'$. For that, we modify $\pi$ in the following way: we only change the images of vertices in $T\cup B$. Namely, $\pi'(v) = \pi(v)$ for $v \in V(G') \setminus (T \cup B)$, and for each $e = uv \in E(G)$,
    \begin{align*}
        \pi'(t_e^1) = \pi(t_e^1) + \theta \cdot \overrightarrow{\pi(u)O},\\
        \pi'(b_e^1) = \pi(b_e^1) + \theta \cdot \overrightarrow{\pi(u)O},\\
        \pi'(t_e^2) = \pi(t_e^2) + \theta \cdot \overrightarrow{\pi(v)O},\\
        \pi'(b_e^2) = \pi(b_e^2) + \theta \cdot \overrightarrow{\pi(v)O},
    \end{align*}
 where $\theta > 0$ is a small value to be defined later.
    We show that the embedding $\pi'$ is indeed a unit ball representation of the complement of $G'$. 

    \begin{claim}
        For every $u, v \in V(G')$, $uv \in E(G')$ if and only if $||\pi'(u) - \pi'(v)|| > D$, for some $D > 2$.
        \label{claim:piprime}
    \end{claim}
    \begin{proof}
        First, we consider distances between the pairs $\pi'(w)$, $\pi'(v)$, where $w \in W$, $v \in T \cup B$.
        Let $v \in T \cup B$ and let $w$ be the unique vertex in $W$ such that $\pi'(v) = \pi(v) + \theta \cdot \overrightarrow{\pi(w)O}$. The respective squared distance is then $||\pi'(v) - \pi'(w)||^2 = R_1^2 + (R_2 + \theta)^2$.
        On the other hand, consider a vertex $w' \in W$ with $w' \ne w$.
        For $v$ and $w'$, the squared distance is $||\pi'(v) - \pi'(w')||^2 \le R_1^2 + R_2^2 + \theta^2 + 2 R_2\theta \cos(\delta/n)$, since the distance is independent in $Ox_1x_2x_3$ and $Ox_4x_5$, and in the latter plane the vertex $w'$ is at least at angle of $\delta / n$ away from the line $\overrightarrow{O\pi(w)}$, which by the law of cosines gives the upper bound above. 

        By setting $D = \sqrt{R_1^2 + R_2^2 + \theta^2 + 2 R_2\theta \cos(\delta/n)}$ we therefore achieve that $||w' - v|| \le D$ for each $w' \ne w$, while $||w - v|| > D$ since 
        $$(R_1^2 + (R_2 + \theta)^2) - (R_1^2 + R_2^2 + \theta^2 + 2 R_2\theta \cos(\delta/n)) = 2 R_2 \theta \cdot (1 - \cos(\delta / n)) > 0.$$
        
        Observe that $2 < D$ since $R_1^2 + R_2^2 = 4$, and that $D < 2 + \xi/2$ for a sufficiently small value of $\theta$; fix $\theta$ so that the latter holds.
        We now verify the distance condition for the remaining pairs.
        First, the distance between vertices of $W$ and $c_1$/$c_2$ is the same under $\pi$ and $\pi'$, so it is at most $2 < D$.
        It remains to consider distances between pairs of vertices in $X$, where $X = T \cup B \cup C$. For each $v, u \in X$, 
        $$||\pi'(u) - \pi'(v)||^2 - ||\pi(u) - \pi(v)||^2 \le 2 \theta^2 R_2^2$$
        since the change from $\pi$  to $\pi'$ shifts each vertex by at most the vector of $\theta \cdot \overrightarrow{\pi(w)O}$ for some $w \in W$; the length of this vector is $\theta R_2$, and $\pi$ acts only into the subspace $Ox_1x_2x_3$ on $X$. Clearly, for sufficiently small $\theta$ we get that
$$-\xi/2 < ||\pi'(u) - \pi'(v)|| - ||\pi(u) - \pi(v)|| < \xi/2.$$
        Therefore, all distances that were at most $2 - \xi$ (at least $2 + \xi$, respectively) under $\pi$ from Claim~\ref{claim:pi} remain at most $2 - \xi / 2$ (at least $2 + \xi / 2$, respectively) under $\pi'$. Since $2 - \xi/2 < 2 < D < 2 + \xi/2$, the proof of the claim is concluded.
    \end{proof}

    By Claim~\ref{claim:piprime}, we get the correctness of the presented reduction.
    It remains to observe that the reduction can be done in polynomial time: Only precision polynomial in input size is required for the parameters used for the computation of the coordinates.
    Since in the resulting instance of \probcliqueparition there are $O(n + m)$ vertices, and \textsc{$3$-Coloring} does not admit a $2^{o(n + m)}$-time algorithm under the \classETH, the statement of the theorem follows.
\end{proof}


\bibliography{ref}

\end{document}